\documentclass[aps,prl,twocolumn,superscriptaddress,showpacs,floatfix]{revtex4-1}
\usepackage[english]{babel}
\usepackage{graphicx}
\usepackage{stmaryrd}
\usepackage{amssymb}
\usepackage{amsfonts}
\usepackage{amsmath}
\usepackage{dsfont}
\usepackage{color}
\usepackage{xspace}
\usepackage{multirow}
\usepackage{comment}

\usepackage{bm}
\usepackage{amsthm}

\DeclareSymbolFont{largesymbols}{OMX}{cmex}{m}{n}

\newcommand{\calH}{{\mathcal H}}

\newcommand{\calK}{{\mathcal K}}
\newcommand{\calL}{{\mathcal L}}
\newcommand{\kappamax}{\kappa_{\mbox{\scriptsize max}}}
\newtheorem{theorem}{Theorem}

\DeclareGraphicsExtensions{.pdf,.ps,.eps}

\begin{document}

\title{
Ground-state Energies of Spinless Free Fermions and Hard-core Bosons}
\author{Wenxing Nie}
\affiliation{Institute for Solid State Physics, University of Tokyo,
Kashiwa 277-8581, Japan}
\author{Hosho Katsura}
\affiliation{Department of Physics, Gakushuin University, Tokyo 171-8588, Japan}
\author{Masaki Oshikawa}
\affiliation{Institute for Solid State Physics, University of Tokyo,
Kashiwa 277-8581, Japan}
\date{\today}

\begin{abstract}
We compare the groundstate energies of bosons and fermions with the
same form of the Hamiltonian. If both are noninteracting,
the groundstate energy of bosons is always lower, owing to Bose-Einstein
Condensation. However, the comparison is nontrivial when bosons
do interact. We first prove that, when the hopping is unfrustrated
(all the hopping amplitudes are non-negative), hard-core bosons
still must have a lower groundstate energy than fermions.
If the hopping is frustrated, bosons can have higher groundstate
energy than fermions. We prove rigorously that
this inversion indeed occurs in several examples.
\end{abstract}


\maketitle

\textit{Introduction. ---}
Statistics of identical particles is one of the most fundamental
concepts in quantum physics.
A consequence of the particle statistics appears
in the groundstate (GS) energy.
For a system of free particles, the GS of bosons is
obtained by putting all the particles in the lowest-energy state
of the single-particle Hamiltonian, while the GS of fermions is
obtained by filling the individual single-particle states up to the Fermi level (Pauli exclusion principle).
Thus, the GS energy of bosons $E_0^{\rm B}$  and
that of fermions $E_0^{\rm F}$, for the same form of the Hamiltonian, satisfy
\begin{equation}
E_0^{\rm B}\leq E_0^{\rm F},
\label{eq.e0b.leq.e0f}
\end{equation}
if the particles are non-interacting.

The comparison becomes nontrivial
when the particles do interact;
the Bose-Einstein condensation
is no longer perfect in interacting systems.
Intuitively, it would be still natural to expect
that Eq.~\eqref{eq.e0b.leq.e0f} holds.
However, recently
an apparent counterexample was found
numerically~\cite{HuberAltman2010,Altman-priv}.
This motivates us to examine the fundamental question:
how general is the ``natural'' inequality~\eqref{eq.e0b.leq.e0f}
and when can it be actually violated?

To simplify the matter, in this paper we focus on the comparison
of hard-core bosons with spinless free fermions.
(See also Refs.~\cite{Henley1,Henley2}.)
The Hamiltonian is given by
\begin{equation}
\calH=-\sum_{j,k}\left(
t_{jk}c^{\dagger}_jc_k+\mbox{H.c.}
\right)
-\sum_j\mu_jn_j
+\sum_{j,k}V_{jk}n_jn_k,
\label{eq.Ham}
\end{equation}
where each site $j$ belongs to a lattice $\Lambda$,
$n_j\equiv c^{\dagger}_jc_j$, and $t_{jk}=0$ is assumed for $j=k$.
The uniform (site independent) part of $\mu_j$ is the chemical
potential $\mu$.
For a system of fermions (bosons), we identify $c_j$ with the fermion (boson)
annihilation operator $f_j$ ($b_j$) satisfying the standard anticommutation (commutation) relations.
For bosons, the hard-core constraint ($n_j = 0,1$) may be implemented by introducing
the on-site interaction $U\sum_jn_j(n_j - 1)$ and then
taking $U\to\infty$.

We note that the Hamiltonian~\eqref{eq.Ham} conserves
the total particle number.
Thus the GS can be defined for
a given number of particles $M$ (canonical ensemble),
or for a given chemical potential $\mu$ (grand canonical ensemble).
The comparison between bosons and fermions can be made in
either circumstance.

\textit{Natural Inequality. ---}
First we present a sufficient condition for the
``natural'' inequality~\eqref{eq.e0b.leq.e0f} to hold.
(See Ref.~\cite{Becca_proofbyPF} for a similar inequality for
spinful fermions.)
Furthermore, we find a sufficient condition for the
strict inequality $E_0^{\rm B}  < E_0^{\rm F}$ to hold.
The proof also gives us a physical insight into the reason why
the inequality still holds even in interacting systems,
where the simple argument based on a perfect
Bose-Einstein condensation of bosons breaks down.
\begin{theorem}
The inequality~\eqref{eq.e0b.leq.e0f} holds for any
given number of particles $M$ on a finite lattice $\Lambda$ with
$N \geq M$ sites,
if all the hopping amplitudes $t_{jk}$ are real and non-negative.
Furthermore,
if the lattice $\Lambda$ is connected and has a site connected
to three or more sites, and if the number of particles satisfies
$2 \leq M \leq N-2$, the strict inequality $E_0^{\rm B} < E_0^{\rm F}$ holds.
\label{thm.natural}
\end{theorem}

\begin{proof}
Let us take the occupation number basis
$|\phi^a\rangle\equiv|\{n^a_j\}\rangle$,
where $\sum_j n^a_j = M$.
The number operator $n_j$ has the same matrix elements
in this basis, for hard-core bosons and spinless fermions.
It is convenient to introduce the operator
$\calK^{\rm B,F} \equiv - \calH^{\rm B,F} + C \mathds{1}$ with
a sufficiently large $C$ so that all the eigenvalues
and thus all the diagonal matrix elements $\calK^{\rm B,F}_{aa}$
are positive.
The matrix elements of each hopping term in
the bosonic operator $\calK^{\rm B}$ is non-negative, while
the corresponding matrix element for the fermionic operator
must have the same absolute value but could differ in sign.
Thus the matrix elements for bosonic and fermionic operators satisfy
\begin{equation}
\label{eq.Kab}
\calK_{ab}^{\rm B}=\left\{
\begin{array}{ll}
|\calK_{ab}^{\rm F}|&(a\ne b)\\
\calK_{aa}^{\rm F}& (a=b)
\end{array}
\right.
\end{equation}
The GS of the Hamiltonian $\calH^{\rm B,F}$ corresponds to the
eigenvector belonging to the largest eigenvalue $\kappamax^{\rm B,F}$
of $\calK^{\rm B,F}$.
Let $|\Psi_0\rangle_{\rm F}=\sum_a\psi_a|\phi^a\rangle_{\rm F}$
be the normalized GS for fermions.
The trial state for the bosons
$|\Psi_0\rangle_{\rm B}=\sum_a|\psi_a| |\phi^a\rangle_{\rm B}$,
where $|\phi^a\rangle_{\rm B}$ is the basis state for bosons
corresponding to $|\phi^a\rangle_{\rm F}$.
Then, by a variational argument, Eq.~\eqref{eq.Kab} implies
$\kappamax^{\rm B} \geq \kappamax^{\rm F}$, which is nothing but
the first part of Theorem 1.
As a simple corollary, the GS energies for
a given chemical potential $\mu$
also satisfy Eq.~\eqref{eq.e0b.leq.e0f}.

Let us now consider $\calL^{S} \equiv \left(\calK^S\right)^n$,
where $S={\rm B,F}$, for a positive integer $n$.
Its matrix elements in the
occupation number basis can be expanded as
\begin{equation}
\calL^{S}_{ab} = \sum_{c_1,\ldots,c_{n-1}}
\calK^{S}_{ac_1}\calK^{S}_{c_1c_2}\calK^{S}_{c_2c_3}\ldots
\calK^{S}_{c_{n-1}b} .
\label{eq.L_elem}
\end{equation}
Each term in the sum represents a process in which
a particle can hop to a connected site.
When the lattice $\Lambda$ is connected, any basis state
$|\phi^a \rangle_{\rm B}$ can be reached by a consecutive
application of the hopping term in $\calK^{\rm B}$, and thus
the matrix $\calK^{\rm B}_{ab}$ satisfies the connectivity.
Together with the property $\calK_{ab}^{\rm B} \geq 0$,
$\calK^{\rm B}_{ab}$ (and thus also $\calL^{\rm B}_{ab}$)
is a Perron-Frobenius matrix~\cite{horn2012matrix}.

\begin{figure}
\includegraphics[width=0.8\columnwidth]{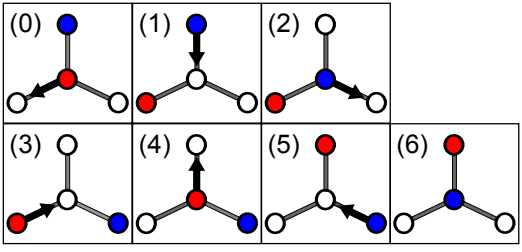}
\caption{An illustration of two-particle exchange process in six steps.}
\label{fig.particle-exchange}
\end{figure}

In contrast to the bosonic system, matrix elements
and thus the amplitude for each process can be
negative for the fermionic system.
In particular, we consider the following situation.
When the lattice has a ``branching'' site connected
to three or more sites and if $2 \leq M \leq N-2$, there is
an initial state $|\phi^a \rangle$ from which
one can exchange two particles and come back to
the initial state in $6$ hoppings
as in Fig.~\ref{fig.particle-exchange}.
Therefore, there is always a nonvanishing negative contribution
to the diagonal element for fermion $\calL^{\rm F}_{aa}$ at $n=6$, while
the corresponding contribution for bosons is positive.
This implies the strict inequality $\calL^{\rm B}_{aa} > \calL^{\rm F}_{aa}$ for
the particular diagonal element.
Applying a corollary of Perron-Frobenius
theorem~\footnote{See Theorem 8.4.5 of Ref.~\cite{horn2012matrix}.}
we find $\kappamax^{\rm B} > \kappamax^{\rm F}$ and hence
the latter part of the theorem follows.
\end{proof}

The sign of hopping amplitudes in a system of bosons may be
related to \emph{frustration}, by mapping
the system of hard-core bosons to a quantum spin system
with $S=1/2$.
Positive hopping amplitudes $t_{jk}$ correspond to
ferromagnetic XY interaction; the corresponding spin system
is an unfrustrated ferromagnet, if all the amplitudes $t_{jk}$
are non-negative.
Theorem~\ref{thm.natural} means that,
if there is no frustration among the hoppings in this sense,
hard-core bosons always have lower energy than the corresponding fermions.
In fact, we can understand this result as an effect of frustration
induced by fermionic statistics, in the following sense.

The many-body problem defined by $\calH^{\rm B,F}$ can be mapped to
a single-particle tight-binding problem if we identify the
many-body basis state $|\phi^a\rangle$ with a site $a$
in a fictitious lattice.
For the boson problem, all the hopping amplitudes in this
single-particle problem are again non-negative.
These are non-frustrating since there is a
constructive interference among all the paths.
In this picture, the fermionic statistics has an effect
of introducing the phase in the hopping.
In particular, when the two-particle exchange can occur in
the original many-body problem, there is a loop in the
fictitious lattice which contains $\pi$ flux.
This could be interpreted as a frustration, since
there is a destructive interference among different
paths. Indeed this is the key property exploited in our proof.
By the mapping to the single-particle problem on the fictitious
lattice, Theorem~\ref{thm.natural} may be regarded as
a particular case of the diamagnetic
inequality on the lattice
(See the paper~\cite{diamagnetic_inequality} and references therein).
In fact, we have also proved a strict version of the
lattice diamagnetic inequality, which has not been
discussed previously to our knowledge.
We emphasize that, this picture does not rely on the
assumption of a perfect Bose-Einstein condensation
and thus its applicability is not limited to
noninteracting systems of particles.

Let us now discuss how the natural
inequality~\eqref{eq.e0b.leq.e0f} can be violated.
According to Theorem~\ref{thm.natural},
in order to realize the violation,
it is necessary to introduce a frustration
by setting some of the amplitudes $t_{jk}$ negative or complex.
While the presence of frustration is not a sufficient condition,
we will demonstrate that the violation indeed occurs in several
concrete examples.
Intuitively, this means that
we can cancel the effect of
the statistical phases by that of hopping amplitudes,
so that the fermions have lower energy than the corresponding bosons.
In the following, for simplicity, we only discuss
tight-binding models of hard-core bosons and
corresponding noninteracting fermions, setting $V_{jk}=0$.

\textit{Particles on a Ring. ---}
We begin with a simple but instructive example in one dimension:
tight-binding model on a ring
$\calH=-\sum_{j=1}^N
(c_{j}^{\dag}c_{j+1}+\textrm{H.c.})$.
The hard-core boson version of this model,
which is equivalent to the $S=1/2$ XY chain,
can be mapped to the model of non-interacting fermions on the ring
by the Jordan-Wigner transformation~\cite{LSM, S_Katsura}.
Thus the hard-core bosons and fermions are almost equivalent
in this case.

However, a care should be taken on the boundary
condition when we discuss the ring of finite length.
For simplicity,
we assume
the number of sites $N$ is
an integral multiple of $4$, and the number of particles $M=N/2$
(an even integer by assumption).
Then the hard-core bosons with the periodic (antiperiodic)
boundary condition are mapped to noninteracting fermions
with the antiperiodic (periodic, respectively)
boundary condition.
Noninteracting fermions on a ring have
a GS energy density lower by $O(1/N^2)$ for
the antiperiodic boundary condition,
compared to the periodic boundary
condition~\cite{Ginsparg-Applied-CFT,Alcaraz-CFT}.
This implies that the hard-core bosons have lower energy than fermions
on a ring with the periodic boundary condition, conforming to
Theorem~\ref{thm.natural}
since all the hopping amplitudes are non-negative.
On the other hand, the same result implies that,
under the antiperiodic boundary condition,
the hard-core bosons have higher energy than fermions.
Imposing the antiperiodic boundary condition is equivalent to
introducing a $\pi$-flux inside the ring, which
can cancel the effect of the statistical phase so that the
inequality~\eqref{eq.e0b.leq.e0f} is indeed inverted.
The energy difference on the ring, however, vanishes asymptotically
in the thermodynamic limit $N\to\infty$.
Thus, we shall seek for different examples where the
hard-core bosons have higher energy than fermions in the
thermodynamic limit.

\textit{Two-Dimensional System with Flux. ---}
A natural system to consider would be a two-dimensional lattice
with flux.
Under the periodic boundary condition,
the total flux is quantized to integer numbers of flux quanta
(the unit flux quantum $\Phi_0=hc/e$ is $2\pi$ in our unit).
Such a uniform flux can be represented using the
string gauge~\cite{Hatsugai-stringgauge}.
We obtained the GS energy of hard-core bosons
and fermions with various densities of particles and
various values of flux using exact numerical diagonalization, for
square lattices up to 26 sites with periodic boundary conditions.
The result for the $\sqrt{26} \times \sqrt{26}$ square lattice
is shown in Fig.~\ref{fig.gs.sq26}.
\begin{figure}
\includegraphics[width=0.65\columnwidth]{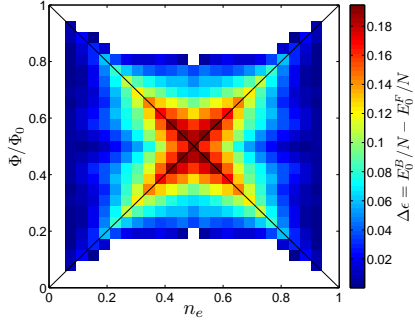}
\caption{Difference of GS energy density $\Delta\epsilon$ between
hard-core bosons and fermions on
the $\sqrt{26}\times\sqrt{26}$ square lattice with $\Phi$
flux per plaquette and $n_e$ particle per site.
The natural inequality~\eqref{eq.e0b.leq.e0f} holds in
white region, while its violation is color coded.
Statistical transmutation is expected along the
two solid diagonal lines.
}
\label{fig.gs.sq26}
\end{figure}

We find that the ``natural'' inequality~\eqref{eq.e0b.leq.e0f}
is violated in a region of the phase diagram.
In particular, the inversion is significant along the
diagonal lines $\Phi/\Phi_0=n_e$ and $\Phi/\Phi_0=1-n_e$.
These lines are precisely where the statistical transmutation
between the hard-core boson and the fermion is expected
to occur~\cite{Semenoff,Fradkin}.
Namely, in the mean-field level, one flux quantum can be attached
to each particle, transforming fermions into bosons and vice versa,
at the same time eliminating the background field.
At zero field, the frustration is absent and hard-core
bosons have lower energy than fermions.
Thus, the statistical transmutation implies that, hard-core bosons
have higher energy than fermions on two diagonal lines.
While this argument is not rigorous and the actual physics
is presumably more involved~\cite{Cooper},
our numerical result
supports the statistical transmutation scenario.
(For a related discussion for spinful electrons, see Ref.~\cite{Saiga}.)
Numerical results for the square lattices of various sizes up to
$26$ sites (not shown) suggest that the energy difference is
nonvanishing in the thermodynamic limit.

In fact, in the following, we will prove rigorously
in the thermodynamic limit that the fermions have
lower energy at half filling with $\Phi=\pi$ flux per plaquette,
as suggested by our numerical calculation and the statistical
transmutation argument.

\textit{Anderson's argument. ---}
\begin{figure}
\includegraphics[width=0.8\columnwidth]{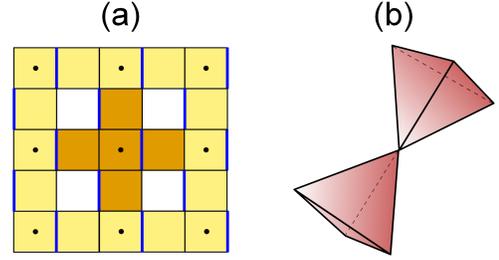}
\caption{
(a) The square lattice with $\pi$ flux in each plaquette.
We choose the gauge so that the hopping amplitude $t_{jk}$
is $+1$ on black links, and $-1$ on blue ones.
The brown cross represents a cluster of $12$ sites.
The whole lattice is covered by clusters, whose centers are denoted by black dots.
(b) A dimer of two tetrahedra made up of 7 sites
in the pyrochlore lattice.
}
\label{fig.piflux}
\end{figure}
Let us discuss the square lattice with $\pi$-flux per plaquette.
The Hamiltonian reads
\begin{equation}
\mathcal H=-\sum_{<j,k>}(t_{jk}c_j^{\dag}c_k+\textrm{H.c.}),
\end{equation}
where $t_{jk}=\pm1$ as specified in Fig.~\ref{fig.piflux}(a).
We note that Lieb has shown that $\pi$-flux minimizes
the GS energy of fermions
at half-filling
on the square lattice~\cite{Lieb-flux}.
On the other hand, an argument similar to the Proof of
Theorem~\ref{thm.natural}
can be used to prove a lattice version of Simon's theorem
on diamagnetism of bosons~\cite{Barry_Simon_diamag}.
Namely, for bosons, introduction of flux always increases
the GS energy.
These, together with the statistical transmutation argument,
suggest a possibility of violation of
Eq.~\eqref{eq.e0b.leq.e0f} with $\pi$-flux per plaquette.

For technical convenience,
we restrict ourselves to the case of
the ``grand canonical ensemble'' GS at $\mu=0$.
For $\pi$-flux square lattice, it turns
out to be equivalent to finding the GS at
half filling ($1/2$ particle per site).
From the exact dispersion relation,
the GS energy of the fermionic model at $\mu=0$
is obtained
exactly as $E_0^{\rm F}\sim-0.958091N$
in the limit of large $N$.
Now we turn to the ``grand canonical'' GS energy,
of the corresponding boson model at the same chemical potential ($\mu=0$).
Here we use Anderson's argument~\cite{Anderson,Tarrach-lower-bound,Valenti-lower-bound}
by writing the Hamiltonian as
\begin{equation}
\calH =\sum_\alpha h_\alpha,
\label{eq.decomp.Anderson}
\end{equation}
where
\begin{equation}
h_\alpha=-\frac{1}{2}\sum_{\langle j,k\rangle\in\boldsymbol{+}_\alpha}
( t_{jk}c_j^{\dag}c_k+\textrm{H.c.}).
\label{eq.cross.cluster}
\end{equation}
Here $\boldsymbol{+}_\alpha$ refers to a cross-shaped cluster
of $12$ sites as shown in Fig.~\ref{fig.piflux}(a).
We consider all the clusters
with the same pattern of hopping amplitudes within the cluster,
in the square lattice.
As a consequence, each cluster as shown in Fig.~\ref{fig.piflux}(a),
overlaps with 4 other clusters and
each link appears in two different clusters
when periodic boundary conditions are imposed.
The factor $1/2$
in Eq.~\eqref{eq.cross.cluster} compensates this double counting.

The GS energy $E_0$ of $\calH$ satisfies
$ E_0\geq\sum_\alpha\epsilon_0^\alpha$,
where $\epsilon_0^\alpha$ is the GS energy of $h_\alpha$.
The grand canonical GS energy of the cross-shaped cluster is obtained
by exact diagonalization as $\epsilon_0^\alpha = - 3.609035$.
Since there are $N/4$ such clusters in the square lattice of $N$ sites,
we obtain
\begin{equation}
E_0^{\rm B}/N\geq-3.609035/4=-0.902259>E_0^{\rm F}/N .
\end{equation}
Thus the inversion of the GS energies for
the $\pi$-flux square lattice model with $\mu=0$,
as expected from the statistical transmutation argument
discussed earlier,
is now proved rigorously.

This argument is not restricted to two-dimensional systems.
Let us consider the standard tight-binding model on
the three-dimensional pyrochlore lattice:
$\calH =\sum_{\langle j,k\rangle}(c^{\dagger}_jc_k+\textrm{H.c.})$,
which has frustrated hoppings with this choice of the sign.
Again we set the chemical potential $\mu=0$.
The model in the single-particle
sector has two degenerate flat bands at the energy 
$\epsilon=-2$
and two dispersive bands touching the flat bands~\cite{Bergman}.
Thus for fermions, the GS energy at $\mu=0$ satisfies
$E_0^{\rm F}<-2(N/2)=-N$, where $N$ is the number of sites
of the lattice.
We note that, because of the lack of the particle-hole symmetry,
$\mu=0$ does not imply half-filling for this model.
The hard-core boson version of this model can be decomposed
as Eq.~\eqref{eq.decomp.Anderson} with
$h_\alpha=(1/4)\sum_{\langle j, k\rangle\in\mbox{\scriptsize TD}_\alpha}
(c^\dagger_jc_k+\textrm{H.c.})$,
where TD$_\alpha$ refers to each dimer of
elementary tetrahedra of the pyrochlore lattice sharing
a vertex (site) (see Fig.\ref{fig.piflux}(b)).
Here we count dimers in any direction; each tetrahedron
(and thus each link) belongs to $4$ dimers.
The factor $1/4$ in the definition of $h_\alpha$ is introduced
to compensate the overcounting.
The GS energy of a tetrahedra dimer is
obtained by exact diagonalization as
$\epsilon_0^\alpha=-(2+\sqrt{2})/4=-0.853554$.
Since there are $N$ dimers of tetrahedra,
the GS energy of bosons at $\mu=0$ satisfies
$E_0^{\rm B}/N\geq-(2+\sqrt{2})/4>E_0^{\rm F}/N$.
Thus we have proved the violation of Eq.~\eqref{eq.e0b.leq.e0f}
for the simple tight-binding model on the three-dimensional
pyrochlore lattice.

The above example of the pyrochlore lattice exhibits a flat
band as the lowest energy band.
While the existence of a flat band is not necessary
to violate Eq.~\eqref{eq.e0b.leq.e0f}, it does tend to help:
as long as all the fermions occupy the lowest flat band,
Pauli exclusion principle plays no role in increasing
the GS energy.
Thus, such flat band models would have a better chance to realize
the inversion of the GS energies.
In fact, we can show that the inequality~\eqref{eq.e0b.leq.e0f}
is indeed violated in a few models with
a lowest flat band in a range of filling fraction,
using a cluster decomposition technique~\cite{prep}.
They include the delta-chain model, for which
the violation of Eq.~\eqref{eq.e0b.leq.e0f} was numerically found for
small clusters~\cite{HuberAltman2010,Altman-priv},
and the kagome lattice model.

\textit{Conclusions. ---}
We have investigated the fundamental question
whether the GS energy of hard-core bosons is
lower than that of fermions on the same lattice.
We have proved that
the former is indeed lower than the latter, as naturally expected,
if there is no frustration in the hopping.
The statistical phase of fermions
induces a magnetic flux in an effective description
in terms of a single-particle problem on a fictitious lattice.
It results in a frustration in the sense of
destructive quantum interferences among different paths.
This also provides a new understanding why the bosons
have lower energy than fermions, when there is no frustration
in the hopping.

On the other hand, the inequality can be reversed in the presence of
frustration, and we have demonstrated that it is actually the case
in several concrete examples.
There is a close connection of the present problem
to many apparently unrelated concepts in quantum many-body physics,
including diamagnetic inequality,
Simon's universal diamagnetism of bosons,
Lieb's optimal flux for fermions, statistical transmutation,
and flat band.
More details of our analysis will be presented in a
separate publication~\cite{prep}.

\medskip

We thank Ehud Altman for stimulating discussion which motivated
us to study this problem.
We also thank Nigel Cooper and Tarun Grover for useful discussions.
This work is supported in part by
Grant-in-Aid for Scientific Research on Innovative Areas
No. 20102008 from MEXT, Japan, and
by US National Science Foundation Grant No. NSF PHY11-25915
during a visit of W.~N. and M.~O.
to Kavli Institute for Theoretical Physics, UC Santa Barbara.
H.~K. was supported in part by Grand-in-Aid
for Young Scientists (B) (Grant No. 23740298).
W.~N. is supported by MEXT scholarship.
Numerical calculations in this work were partially
carried out using TITPACK ver.2 by H. Nishimori, and
codes provided by ALPS project~\cite{ALPS-web,ALPS-ref1,ALPS-diag}.

\bibliography{biblio}

\end{document}